\documentclass[11pt]{article}
\usepackage[utf8]{inputenc}

\usepackage{a4}
\usepackage{amsmath, amssymb, amstext, amsthm,epsfig}
\usepackage{xcolor}

\newcommand{\poly}{\text{poly}}

\newcommand{\cnd}{\mskip 1mu|\mskip 1mu}

\renewcommand{\phi}{\varphi}
\renewcommand{\epsilon}{\varepsilon}
\renewcommand{\ge}{\geqslant}
\renewcommand{\le}{\leqslant}
\usepackage{hyperref}

\theoremstyle{plain}
\newtheorem{theorem}{Theorem}
\newtheorem{lemma}{Lemma}

\newtheorem{proposition}[theorem]{Proposition}

\theoremstyle{remark}
\newtheorem{remark}{Remark}

\DeclareMathOperator*{\pol}{poly}

\usepackage[utf8]{inputenc}

\usepackage{amsthm, amssymb, amsfonts, amsmath}
\usepackage{graphicx}
\newcommand{\KK}{\mathrm{K}}
\newcommand{\KA}{\mathrm{KM}}
\newcommand{\m}{\mathrm{M}}
\begin{document}

\title{Prediction and MDL  for Infinite Sequences}
\author{Alexey Milovanov\footnote{A part of this work was done in National Research University Higher School of Economics, Moscow, Russian Federation. This paper was prepared within the framework of the HSE University Basic Research Program. }
\\ LASIGE, Faculdade de Ciências, Universidade de Lisboa}
\maketitle
\begin{abstract}
We combine Solomonoff's approach to universal prediction with algorithmic statistics and suggest to use the computable measure that provides the best ``explanation'' for the observed data (in the sense of algorithmic statistics) for prediction. In this way we keep the expected sum of squares of prediction errors bounded (as it was for the Solomonoff's predictor) and, moreover, guarantee that the sum of squares of prediction errors is bounded along any Martin-L\"of random sequence.

An extended abstract of this paper was presented at the 16th International Computer Science Symposium in Russia (CSR 2021)~\cite{m}.
\end{abstract}

\section{Introduction}
We consider probability distributions (or measures) on the binary tree, i.e., non-negative functions $P: \{0,1\}^* \to \mathbb{R}$ such that $P(\text{empty word}) = 1$ and $P(x0) + P(x1) = P(x)$ for every string $x$. We assume that all the values $P(x)$ are rational; $P$ is called \emph{computable i}f there exists an algorithm that on input $x$ outputs $P(x)$.

Consider the following prediction problem. Imagine a black box that generates bits according to some unknown computable distribution $P$ on the binary tree. Let $x=x_1\ldots x_n$ be the current output of the black box.  The predictor's goal is to guess the probability that the next bit is $1$, i.e., the ratio $P(1\cnd x)=P(x1)/P(x)$. 

Ray Solomonoff suggested to use the universal semi-measure $\m$ (called also the \emph{a priori probability}) for prediction. Recall that a semi-measure $S$ on the binary tree (a \emph{continuous semi-measure}) is a non-negative function $S: \{0,1\}^* \to \mathbb{R}$ such that $S(\text{empty word}) \le 1$ and $S(x0) + S(x1) \le S(x)$ for every string $x$. Semi-measures correspond to probabilistic processes that output a bit sequence but can hang forever, so an output may be some finite string $x$; the probability of this event is $S(x)-S(x0)-S(x1)$. A semi-measure $S$ is called \emph{lower semi-computable}, or \emph{enumerable}, if the set $\{(x,r) : r<S(x)\}$ is (computably) enumerable. Here $x$ is a string and $r$ is a rational number. Finally, a lower semi-computable semi-measure $\m$ is called \emph{universal} if it is maximal among all semimeasures up to a constant factor, i.e., if for every lower semi-computable semi-measure $S$ there exists $c>0$ such that $\m(x) \ge c S(x)$ for all~$x$. Such a universal semi-measure exists \cite{sol0,suv,LiVit}.\footnote{One may even require that the probabilities for finite outputs, i.e., the differences $S(x)-S(x0)-S(x1)$ are maximal, but we do not require this.} 

Solomonoff suggested to use the ratio $\m(1\cnd x):= \m(x1)/\m(x)$ to predict $P(1\cnd x)$ for an unknown computable measure $P$. He proved the following bound for the prediction errors.  
\begin{theorem}[\cite{sol}]\label{th:sol}
For every computable distribution $P$ and for every $b \in \{0,1\}$ the following sum over all binary strings is finite: 
\begin{equation}
\label{sum}
 \sum_x P(x)\cdot (P(b \cnd x) - \m(b \cnd x))^2 < \infty.   
\end{equation}
Moreover, this sum is bounded by $O(\KK(P))$,  where $\KK(P)$  is the prefix complexity of the computable measure~$P$
(the minimal length of a prefix-free program corresponding $P$). 
\end{theorem}
Note that for semi-measure the probabilities to predict $0$ and $1$ do not sum up to $1$, so the statements for $b=0$ and $b=1$ are not equivalent (but both are true).

The sum from Theorem~\ref{th:sol} can be rewritten as the expected value of the  function $D$ on the infinite binary sequences with respect to $P$, where $D(\omega)$ is defined as
$$
D(\omega)=\sum_{x  \text{ is a prefix of } \omega}
(P(b \cnd x) - \m(b \cnd x))^2.
$$
This expectation is finite, therefore for $P$-almost all $\omega$ the value $D(\omega)$ is finite and
\[
P(b \cnd x) - \m(b \cnd x) \to 0.
\]
when $x$ is an increasing prefix of $\omega$. One would like to have this convergence for all Martin-L\"of random sequences $\omega$ (with respect to measure $P$), but this is not guaranteed, since the null set provided by the argument above may not be an \emph{effectively} null set. An example from~\cite{hm} shows that this is indeed the case.
\begin{theorem}[\cite{hm}]\label{th:hm}
\label{hm}
There exist a specific universal semi-measure $\m$,  computable distribution $P$ and Martin-L{\"o}f random  \textup(with respect to P)  sequence $\omega$ such that
\[
P(b \cnd x) - \m(b \cnd x) \not\to 0.
\]
for increasing prefixes $x$ of $\omega$.
\end{theorem}
Lattimore and Hutter  generalized  Theorem~\ref{hm} by proving the same statement for a wide class of universal semi-measures~\cite{lh}. 

Trying to overcome this problem and get a good prediction for all Martin-L\"of random sequences, we suggest the following approach to prediction. For a finite string $x$ we find a distribution $Q$ on the binary tree that is the best (in some sense) explanation for~$x$. The probabilities of the next bits are then predicted as $Q(0\cnd x)$ and $Q(1\cnd x)$.

This approach combines two advantages. The first is that the series of type \eqref{sum} also converges, though the upper bound for it (at least the one that we are able to prove) is much greater than $O(\KK(P))$. The second property is that the prediction error (defined as in Theorem~\ref{th:hm}) converges to zero for every Martin-L{\"o}f random sequence.

Let us give formal definitions. The quality of the computable distribution $Q$ on the binary tree, considered as an ``explanation'' for a given string $x$, is measured by the value $3\KK(Q) - \log Q(x)$: the smaller this quantity is, the better is the explanation. One can rewrite this expression as the sum
\[
2\KK(Q) + [\KK(Q)-\log Q(x)].
\]
Here the expression in the square brackets can be interpreted as the length of the two-part description of $x$ using $Q$ (first, we specify the hypothesis $Q$ using its shortest prefix-free program, and then, knowing $Q$, we specify $x$ using arithmetic coding; the second part requires about $-\log Q(x)$ bits). The first term $2\KK(Q)$ is added to give extra preference to simple hypotheses; the factor $2$ is needed for technical reasons (in fact, any constant greater than~$1$ will work\footnote{One can consider a similar prediction method with factor $1 + \alpha$ for arbitrary positive $\alpha$ instead of factor $2$. However it does not give an significant improvement so we do not do it. }).

For a given $x$ we select the best explanation that makes this quality minimal.

Formally, we find  $\text{argmin}_Q\{3\KK(Q) - \log Q(x)\}$.

Then we  predict the probability that the next bit after $x$ is~$b$:
\[
H(b \cnd x):= \frac{Q(xb)}{Q(x)},
\]
where $Q$ is the best explanation for string  $x$ (or one of the best explanations if there are several).

In this paper we prove the following results:

\begin{theorem}
\label{sum_solomonoff}
For every computable distribution $P$ the following sum over all binary strings $x$ is finite:
$$
 \sum_x P(x) (P(0 \cnd x) - H(0 \cnd x))^2 < \infty.   
$$
\end{theorem}

\begin{theorem}\label{m-l}

Let   $P$ be a computable measure and let $\omega$ be a Martin-L{\"o}f random sequence with respect to  $P$. Then 
\[
H(0 \cnd x) - P(0\cnd x) \rightarrow 0
\]
for prefixes $x$ of~$\omega$ as the length of prefix goes to infinity.
\end{theorem}

We speak about the probabilities of zeros, but both $P$ and $Q$ are measures, so this implies the same results for the probabilities of ones.

We prove that $$\sum_{\textup{$x$ is a prefix of $\omega$}}  (H(0| x) - P(0| x))^2 < \infty$$
(Theorem~\ref{m-l_sum}) that is the strengthening of Theorem~\ref{m-l}.

\begin{remark}
    Note that function $H$ as well as Solomonoff's predicator is uncomputable. Indeed, consider some computable predictor $H$. Then there exists a computable distribution $P$ such that $|H(0| x) - P(0| x)| \ge \frac{1}{2}$ for every $x$. Of course this predictor does not satisfy Theorem~\ref{sum_solomonoff}.
\end{remark}
In \cite{h2} Hutter suggested a similar approach but without coefficient $3$ for $\KK(Q)$ (see also \cite{h1,h3}). For this approach he proved an analogue of Theorem~\ref{sec:solomonoff} with different proof technique.    

In \cite{hm}  the existence of a semi-computable measure satisfying Theorem~\ref{m-l} was proved. In \cite{hm1} authors prove an analogue of Theorem~\ref{sum_solomonoff} for the Hellinger distance with double exponential in $\KK(P)$ bound.

In the  next section we prove Theorem~\ref{m-l}. 

In Section~\ref{sec:solomonoff} we prove Theorem~\ref{sum_solomonoff}.

Finally, in Section~\ref{enum} we consider the case when we know some information about $P$. More precisely, we know that $P$ belongs to some enumerable set of computable measures. We suggest a similar approach for prediction in this case. We prove analogues of theorems~\ref{m-l} and~\ref{sum_solomonoff} (theorems~\ref{enum_m_l} and ~\ref{sum_solomonoff_enum}) for this prediction method.  We achieved better (polynomial in complexity of $P$) error estimations in these theorems. 

\section{Prediction on Martin-L{\"o}f random sequences}\label{sec:m-l}

Recall the Schnorr--Levin theorem~\cite[ch.5]{suv} that says that a sequence $\omega$ is random with respect to a computable probability measure $P$ if and only if the ratio $\m(x)/P(x)$ is bounded for $x$ that are prefixes of~$\omega$. 

The same result can be reformulated in the logarithmic scale. Let us denote by $\KA(x)$ the \emph{a priori complexity of $x$}, i.e.,   $\lceil-\log \m(x)\rceil$ (the rounding is chosen in this way to ensure upper semicomputability of $\KA$). We have 
\[
\KA(x)\le -\log P(x) + O(1)
\]
for every computable probability measure $P$, where $O(1)$ depends on $P$ but not on $x$. Indeed, since $\m$ is maximal, the ratio $P(x)/\m(x)$ is bounded. Moreover, since $P(x)$ can be included in the mix for $\m(x)$ with coefficient $2^{-\KK(P)}$, we have
\[
\KA(x)\le -\log P(x) + \KK(P)+O(1)
\]
with some constant in $O(1)$ that does not depend on $P$ (and on $x$). As we have discussed in the previous section, the right-hand side includes the length of the two-part description of $x$.

Let us call 
\[
d(x\cnd P):=-\log P(x)-\KA(x)
\]
the \emph{randomness deficiency} of a string $x$ with respect to a computable measure $P$. (There are several notions of deficiency, but we need only this one.) Then we get
\[
d(x\cnd P)\ge -\KK(P)-O(1)
\]
so the deficiency is almost non-negative. The Schnorr--Levin theorem characterizes Martin-L\"of randomness in terms of deficiency:

\begin{theorem}[Schnorr--Levin]
\leavevmode

\textbf{\textup{(a)}} If a sequence $\omega$ is Martin-L\"of random with respect to a computable disctibution $P$, then $d(x\cnd P)$ is bounded for all prefixes $x$ of $\omega$.

\textbf{\textup{(b)}} Otherwise (if $\omega$ is not random with respect to $P$), then $d(x\cnd P)\to\infty$ as the length  of a prefix $x$ of $\omega$ increases.

\end{theorem}

Note that there is a dichotomy: the values $d(x\cnd P)$ for prefixes $x$ of $\omega$ either are bounded or converge to infinity (as the length of $x$ goes to infinity). We can define randomness deficiency for infinite sequence $\omega$ as
$$
d(\omega\cnd P) := \sup_{x \text{ is prefix of } \omega} d(x\cnd P);
$$
it is finite if and only if $\omega$ is random with respect to $P$.

Let us also recall the following result of Vovk:

\begin{theorem}[\cite{Vovk}]
\label{vovk}
Let $P$ and $Q$ be two computable distributions. Let $\omega$ be a Martin-L\"of random sequence with respect both to $P$ and $Q$.  Then
\[
P(0\cnd x) - Q(0\cnd x)\to 0
\]
for prefixes $x$ of~$\omega$ as the length of prefix goes to infinity.
\end{theorem}

We will prove this theorem (and even more exact statement) in the next section.


\begin{proof}[Proof of Theorem~\ref{m-l}]

Now we have a sequence $\omega$ that is Martin-L\"of random with respect to some computable measure $P$, so $D=d(\omega\cnd P)$ is finite.  For each prefix $x$ of $\omega$ we take the best explanation $Q$ that makes the expression
\[
3\KK(Q)-\log Q(x)
\]
minimal. Note that $P$ is among the candidates for $Q$, so this expression should not exceed
\[
3\KK(P)-\log P(x).
\]
Since $\omega$ is random with respect to $P$ and $x$ is a prefix of $\omega$, Schnorr--Levin theorem guarantees that the latter expression 
\[
3\KK(P)-\log P(x)=\KA(x)+O(1)
\]
where constant in $O$ depends on $P$ and $\omega$ but not on~$x$. On the other hand, the inequality $\KA(x)\le \KK(Q)-\log Q(x)+O(1)$ implies that
\begin{equation}
\label{first}
3\KK(Q)-\log Q(x) = 2\KK(Q) + \KK(Q) - \log Q(x) \ge 2\KK(Q)+\KA(x) - O(1).     
\end{equation}
 
So measures $Q$ with large $\KK(Q)$ cannot compete with $P$, and there is only a finite list of candidate measures for the best explanation $Q$. For some of these $Q$ the sequence $\omega$ is $Q$-random with respect to $Q$, so one can use Vovk's theorem to get the convergence of predicted probabilities when these measures are used.

Still we may have some ``bad'' $Q$ in the list of candidates for which $\omega$ is not $Q$-random. However, the Schnorr--Levin theorem guarantees that for a bad $Q$ we have
\[
-\log Q(x) - \KA(x) \to\infty
\]
if $x$ is a prefix of $\omega$ of increasing length. So the difference between two sides of~\eqref{first} goes to infinity as the length of $x$ increases, so $Q$ loses to $P$ for large enough $x$ (is worse as an explanation of $x$). Therefore, only good $Q$ will be used for prediction after sufficiently long prefixes, and this finishes the proof of Theorem~\ref{m-l}.
\end{proof}


\section{On the expectation of squares of errors} 

\label{sec:solomonoff}
In this section we prove Theorem~\ref{sum_solomonoff}.
First we will prove some strengthening of Theorem~\ref{vovk}
\begin{lemma}
\label{lem1}
Let $P$ and $Q$ be computable distributions. and let $\m$ be a universal semi-measure. 
Assume that for string $ x=x_1  \ldots x_n$ and $C > 0$ it holds that $P(x),Q(x) \ge  \m(x)/C$. Then:
$$\sum_{i=1}^{n-1} (P(x_i | x_1 \ldots x_{i-1}) - Q(x_i | x_1 \ldots x_{i-1}))^2 = O(\log C + \KK(P,Q)).$$
\end{lemma}

\begin{proof}[Proof of Theorem~\ref{vovk} from Lemma~\ref{lem1}]
According to one of definitions of Martin-L{\"o}f randomness the values $\m(x) / P(x) $ and $\m(x) / Q(x)$ are bounded by a constant. It reminds to use Lemma~\ref{lem1}.
\end{proof}

\begin{proof}[Proof of Lemma~\ref{lem1}]
Denote
\[
p_i = P(x_i \cnd x_1\ldots x_{i-1}), \quad 
q_i = Q(x_i \cnd x_1\ldots x_{i-1}).
\] Note that
\[
P(x_1\ldots x_n)= p_1 p_2\ldots p_n, \quad 
Q(x_1\ldots x_n)= q_1 q_2\ldots q_n.
\]
Now consider the ``intermediate'' measure $R$ for which the probability of $0$ (or $1$) after some $x$ is the average of the same conditional probabilities for $P$ and $Q$:
\[
R(0\cnd x_1\ldots x_{i-1}) =\frac{P(0\cnd x_1\ldots x_{i-1})+Q(0\cnd x_1\ldots x_{i-1})}{2}. 
\]
The corresponding $r_i=R(x_i\cnd x_1\ldots x_{i-1})$ are equal to $(p_i+q_i)/2$. 

Probability distribution $R$ is computable and $\KK(R) \le \KK(P,Q) + O(1)$. Hence,  it holds that $R(x) \le 2^{\KK(P,Q)} \m(x) \le 2^{\KK(P,Q)} \cdot C\cdot P(x)$. The similar inequality holds for distribution $Q$. Therefore:
$$r_1 \cdots r_n \le C  \cdot 2^{\KK(P,Q)} \cdot p_1 \cdots p_n $$
and
$$r_1 \cdots r_n \le C  \cdot 2^{\KK(P,Q)} \cdot q_1 \cdots q_n .$$

Multiplying we obtain:
\begin{equation}
\label{rev}
(\frac{p_1 + q_1}{2}\cdots \frac{p_n + q_n}{2})^2  \le C^2 \cdot 2^{2\KK(P,Q)} \cdot p_1 \cdots p_n \cdot q_1 \cdots q_n .
\end{equation}

These two inequalities show that the product of arithmetical means of $p_i$ and $q_i$ is not much bigger than the product of their geometrical means, and this is only possible if $p_i$ is close to $q_i$ (logarithm is a strictly convex function).

To make the argument precise, recall the bound for the logarithm function:
\begin{lemma}\label{lem:convex}
For $p,q\in (0,1]$ we have
\[
\log \frac{p+q}{2}- \frac{\log p + \log q}{2} \ge \frac{1}{8\ln 2}(p-q)^2
\]
\end{lemma}

\begin{proof}
Let us replace the binary logarithms by the natural ones; then the factor $\ln 2$ disappears. Note that the left hand side remains the same if $p$ and $q$ are multiplied by some factor $c\ge 1$ while the right side can only increase. So it is enough to prove this for $p=1-h$ and $q=1+h$ for some $h\in(0,1)$, and this gives
\[
-\frac{\ln(1-h) + \ln(1+h)}{2}\ge \frac{1}{2}h^2;
\]
and this happens because $\ln(1-h)+\ln(1+h)=\ln(1-h^2)\le -h^2$.
\end{proof}
For the product of $n$ terms we get the following bound:
\begin{lemma}\label{lem:convexn}
If for $p_1,\ldots,p_n,q_1,\ldots,q_n\in (0,1]$ we have
\[
\left(\frac{p_1+q_1}{2}\cdot\ldots\cdot\frac{p_n+q_n}{2}\right)^2 \le c p_1\ldots p_nq_1\ldots q_n,
\]
then $\sum_i (p_i-q_i)^2 \le O(\log c)$, with some absolute constant hidden in $O(\cdot)$-notation.
\end{lemma}
\begin{proof}
Taking logarithms, we get 
\[
2\sum_i \log\frac{p_i+q_i}{2} \le \log c + \sum_i \log p_i + \sum_i \log q_i,
\]
and therefore
\[
\sum_i \left(\log\frac{p_i+q_i}{2}-\frac{\log p_i + \log q_i}{2}\right) \le \frac{1}{2}\log c.
\]
It remains to use Lemma~\ref{lem:convex} to get the desired inequality.
\end{proof}
To complete the proof of Lemma~\ref{lem1} it remains to use in inequality \eqref{rev} Lemma~\ref{lem:convexn} for   $c:=C^2 \cdot 2^{2\KK(P,Q)}$.
\end{proof}

Now we  prove a strengthening of Theorem~\ref{m-l}. 

\begin{theorem}
\label{m-l_sum}
Let   $P$ be a computable measure, let  $\omega$ be a Martin-L{\"o}f random sequence with respect to  $P$ such that $d(\omega|P)= D$.

Then $$\sum_{\textup{$x$ is a prefix of $\omega$}}  (H(0| x) - P(0| x))^2 = O((\KK(P) + D)\cdot 2^{\frac{3 \KK(P) + D + O(1)}{2}}).$$ 
\end{theorem}
\begin{proof}

Assume that distribution $Q$ is the best for some $x =x_1\ldots x_n$. 
Then 
\begin{equation}
\label{eq1}
  3 \KK(Q) - \log Q(x) \le 3 \KK(P) - \log P(x).  
\end{equation}

Since $d(\omega |P)=D$ we obtain that 
\begin{equation}
\label{eq2}
 -\log P(x) \le \KA(x) +D.   
\end{equation}
Therefore,
$$- \log Q(x) \le 3 \KK(P) - \log P(x) \le 3 \KK(P) + \KA(x) +D \text{ , so}$$
$$ Q(x) \ge M(x) \cdot 2^{-3\KK(P)- D} \text{ and}$$
$$ P(x) \ge M(x) \cdot 2^{- D}.$$
We want to estimate $\sum_{i=1}^n (Q(0| x_1 \ldots x_i)- P(0| x_1 \ldots x_i))^2$   by  Lemma~\ref{lem1}.
We can use this lemma for $C=2^{3\KK(P)+ D}$.

From \eqref{eq1}, \eqref{eq2} and  $\KA(x)\le -\log Q(x)+\KK(Q)$ it follows that 
\begin{equation}
\label{eq3}
\KK(Q) \le \frac{3 \KK(P) + D + O(1)}{2}.
\end{equation}
Therefore by Lemma~\ref{lem1} we obtain
$$\sum_{i=1}^{n-1} (Q(0| x_1 \ldots x_i) - P(0| x_1 \ldots x_i))^2 = O(\KK(P) + D).$$
In fact, we can not use this lemma for the last term $(Q(0| x) - P(0| x))^2$. This term we just bound by  $1$. 

So, every probability distribution that is the best for some $x$ ``contributes'' $O(\KK(P) + D)$ in the sum $\sum_{x \text{ is a prefix of }\omega}
(H(0\cnd x) - P(0 \cnd x))^2$.

There are at most $2^{\frac{3 \KK(P) + D + O(1)}{2}}$ such distribution by \eqref{eq3}, so we obtain the required estimation. 
\end{proof}

Recall the following well-known statement
\begin{proposition}
\label{measure}
Let $P$ be a computable distribution. Then the $P$-measure of all sequences $\omega$ such that $d( \omega \cnd P) \ge D$ is not greater than $2^{-D}$.   
\end{proposition}
\begin{proof}[Proof of Theorem~\ref{sum_solomonoff}]
Denote by $\Omega$  the set of all infinite sequences with zeros and ones. Note that 

$$
 \sum_x P(x) (P(0 \cnd x) - H(0 \cnd x))^2=  \int_{(\Omega, P)} \sum_{\textup{$x$ is a prefix of $\omega$}} (H(0| x) - P(0| x))^2.   
$$
By Theorem~\ref{m-l_sum} we can estimate  the sum in the integral for  sequence $\omega$ with $d(x|\omega)=D$ as  $O((\KK(P) + D)\cdot 2^{\frac{3 \KK(P) + D + O(1)}{2}})$. By Proposition~\ref{measure} the measure of sequences with such randomness deficiency is at most $2^{-D}$. So we can estimate the integral as $$\sum_{D=0}^{\infty}  O((\KK(P) + D)\cdot 2^{\frac{3 \KK(P) + D + O(1)}{2}}) 2^{-D} = O(\KK(P) 2^{\frac{3 \KK(P)}{2}} ).$$  
(Recall that  the $P$-measure of sequences that are not Martin-\L{\"o}f random with respect to $P$ is equal to $0$, so they do not affect to the integral.)
\end{proof}

\section{Prediction for enumerable classes of hypotheses}
\label{enum}
Assume that we have some information about distribution $P$. We know that $P$ belongs to some enumerable set $\mathcal{A}$ of computable distributions, (i.e. there is an algorithm that enumerate programs that generate distributions from $\mathcal{A}$). For this case it is natural to consider the following measure of complexity  in $\mathcal{A}$: 
$$\KK_{\mathcal{A}}(P):= \KK (i_P), \text{ 
 where } i_p  \text{ is the number of }  P  \text{ in a computable enumeration of }  \mathcal{A}. $$  

 If $P$ appears multiple times in an enumeration, we select $i_P$ with the lowest complexity.
(This definition depends on the choice of a computable enumeration, but the influence of this dependence is constrained by an additive constant.)
 Clearly, $ \KK_\mathcal{A}(P) \ge \KK(P) + O(1)$.

 We can now extend our prediction method. To predict the next bit of $x$, we choose $Q \in \mathcal{A}$ with the lowest value of $3 \KK_\mathcal{A}(Q) - \log Q(x)$ and base our next-bit prediction on $Q$
 : $$H_{\mathcal{A}}(x):= \frac{Q(xb)}{Q(x)}.$$
 
In this section, we demonstrate that if set $\mathcal{A}$ possesses certain favorable properties, analogous results to previous theorems emerge. Furthermore, we can achieve an improved error estimation.
 
 We assume that enumerable set $\mathcal{A}$ has the following property:
 if $P_1, \ldots, P_k \in \mathcal{A}$ then their mixture $\frac{P_1 + \ldots P_k}{k}$ belongs to $\mathcal{A}$.    Moreover there exists an algorithm that for given numbers of $P_1, \ldots, P_k$ (in some enumeration of $\mathcal{A}$) outputs the number of their mixture.
 
 Further everywhere $\mathcal{A}$  is an enumerable set of computable distributions with this property.
 \begin{remark}
 Consider the following example of set $\mathcal{A}$: the set of all \emph{provable} (in some proof system) computable distributions on the binary tree. For every program $p \in \mathcal{A}$, there exists a proof that $p(x)$ halts for every $x$, $p(x) = p(x0) + p(x1)$, and $p(\text{empty word}) = 1$.
 We guess that all using in practice computable distributions are provable computable, so, in some sense we  get better error estimation ``almost  free''. Our discussion about practice might look unsuitable because our prediction method is not computable. However,  it can be considered as the ideal version of the MDL principle prediction which is approximated in practice.      
 \end{remark}
 
\begin{theorem}
\label{enum_m_l}
Let   $P \in \mathcal{A}$ be a computable measure, let  $\omega$ be a Martin-L{\"o}f random sequence with respect to  $P$ such that $d(\omega|P)= D$.

Then $$\sum_{\textup{$x$ is a prefix of $\omega$}}  (H_{\mathcal{A}}(0| x) - P(0| x))^2 = O((\KK_{\mathcal{A}}(P) + D)\cdot \pol(\KK_{\mathcal{A}}(P) + D).$$ 
\end{theorem}

\begin{theorem}
\label{sum_solomonoff_enum}
For every computable distribution $P \in \mathcal{A}$ the following sum over all binary strings $x$ is finite:
$$con
 \sum_x P(x) (P(0 \cnd x) - H(0 \cnd x))^2 < \pol(\KK_{\mathcal{A}}(P)).   
$$
\end{theorem}
The proofs of these theorems are in general the same as the proofs of theorems~\ref{m-l_sum} and~\ref{sum_solomonoff}, however some new tools are added. The difference is that we can get better estimation on the number of possible best explanations for prefixes of some sequence.   
\begin{lemma}
\label{enum_lemma}
Let $x$ be a finite string. Assume that there are $2^k$ probability distributions  

$Q_1, \ldots Q_{2^k} \in \mathcal{A}$ such that for every $i$ it holds $\KK_\mathcal{A}(Q_i) \le a$ and $Q_i(x) \ge 2^{-b}$. Then there is probability distribution $Q \in \mathcal{A}$ such that 

$\KK_\mathcal{A}(Q) \le a - k + O(\log (a + k))$ and $Q(x) \ge 2^{-b-k}$. 
\end{lemma}
Note that  $3 \KK_\mathcal{A}(Q) - \log Q(x) \le 3 \cdot (a - k) + O(\log (a + k)) + b + k \le 3 \cdot a - b$ for big enough $k$. This means that string $x$  can not has many ``best'' explanations.  
\begin{proof}[Proof of Lemma~\ref{enum_lemma}]
Let enumerate all distributions of $\mathcal{A}$ with complexity at most $a$ by groups of size $2^{k-1}$ (the last group can be incomplete). The number of such groups is $O(2^{a-k})$. The complexity of every group is at most $a - k + O(\log (a + k))$. Indeed, to describe a group we need its ordinal number in an enumeration and describe this enumeration (we need to know $k$, $a$ and some enumeration of $\mathcal{A}$).

One of these complete group contains some $Q_i$.
Define $Q$ as the mixture of the distributions in this group. 
Since the group has complexity at most $a - k + O(\log a + k)$ the same estimation holds for the complexity of $Q$.  Since some $Q_i$ belongs to the mixture it holds that $Q(x) \ge 2^{-b-k+1}$. 
Recall that $Q$ belongs to $\mathcal{A}$ because every mixture of distributions from $\mathcal{A}$ belongs to $\mathcal{A}$.
\end{proof}
Also we need the following lemma.
\begin{lemma}
\label{prefix}
Let string $s$ be a prefix of string $h$ and let $P$ be a computable distribution such that $d(s \cnd P) =D$. Then $d(h \cnd P) \ge  D  - 2 \log D + O(1)$.
\end{lemma}
(So,  a prefix oconf a string that has small deficiency, has
(almost as) small deficiency).

In fact the proof of this lemma is the same as the proof of Theorem 124 in \cite{suv}.  
\begin{proof}[Proof of Lemma~\ref{prefix}]
For each $k$ consider the enumerable set of all finite sequences that have
deficiency greater than $k$. All the infinite continuations of these sequences form an
open set $S_k$, and $P$-measure of this set does not exceed $2^{-k}$.
 Now consider the
measure $P_k$ on $\Omega$ that is zero outside $S_k$ and is equal to $2^kP$ inside $S_k$. That means
that for every set $U$ the value $P_k(U)$ is defined as $2^k P(U \cap P_k)$. Actually, $P_k$ is
not a probability distribution according to our definition, since $P_k(\Omega)$ is not equal to $1$. However,
$P_k$ can be considered as a lower semi-computable semi-measure, if we change it a bit
and let $P_k(\Omega)=1$ (this means that the difference between $1$ and the former value
of $P_k(\Omega)$ is assigned to the empty string).

It is clear that $P_k$ is lower semi-computable since $P$ is computable and $S_k$ is enumerable.

Now consider the sum

$$ S =\sum_k \frac{1}{k (k-1)}P_k $$
It is a lower semi-computable semi-measure  again
 Then we have
 $$-\log S(x) \le -\log P(x) - k +2\log k + O(1) $$
for every string $x$ that has a prefix with deficiency greater than $k$. Since $S$ does not
exceed  a priori probability (up to $O(1)$-factor), we get the desired
inequality.
\end{proof}

\begin{proof}[Proof of Theorem~\ref{enum_m_l}]

{\bf Part $1$}.
We claim that there are only $\poly(D + \KK_{\mathcal{A}}(P))$ different distributions that are the best for some prefix of $\omega$. 

Let $x$ be a prefix of $\omega$ and $Q$ is the best distribution for $x$. As in the proof of Theorem~\ref{m-l_sum} (see \eqref{eq3})  we obtain 
\begin{equation}
\label{en1}
 \KK_{\mathcal{A}}(Q) \le \frac{3 \KK_{\mathcal{A}}(P) + D + O(1)}{2},   
\end{equation}
$$Q(x) \ge M(x)\cdot 2^{-3\KK_{\mathcal{A}}(P)-D}$$
and hence
\begin{equation}
\label{en2}
d(x \cnd Q)\le 3 \KK_{\mathcal{A}}(P) +D.    
\end{equation}

Let $Q_1, \ldots, Q_m$ be different and the best distribution for prefixes $x_1, \ldots, x_m$ of $\omega$. 

We need to prove that $m=\poly(D + \KK_{\mathcal{A}}(P))$. 

Fix some natural $a$ and $b$. Consider all $Q_i$ such that
 $\KK(Q_i) = a$ and $$b \le d(x_i \cnd Q_i) < 2b.$$ 

It is enough to  we prove that there are only $\poly(D + \KK(P))$ best distributions with the such complexity and randomness deficiency.  Indeed, the honest estimation of $m$ will be multiplied by $\poly(D + \KK(P))$ because these parameters are bounded by polynomials by \eqref{en1} and \eqref{en2}. So, further we consider only $Q_i$ with such parameters.

Let $x_i$ be the shortest prefix among $x_1, \ldots x_m$.

By Lemma~\ref{prefix}  every $Q_j$ is ``rather good'' distribution for $x_i$: $d (x_i \cnd Q_j) \le 2b + O(\log b)$ and hence 
$Q_j(x_i) \ge Q_i(x) \cdot 2^{- O(\log b)}$.
By Lemma~\ref{enum_lemma} there exists a distribution from $R \in \mathcal{A}$ such that $$\KK_{\mathcal{A}}(R) \le a - \log m + O(\log a + \log m) \text{ and}$$  $$R(x) \ge Q_i(x) \cdot 2^{- \log m - O(\log b)}.$$
Since $Q_j$ is not worse distribution then $R$ for $x$ we have:
$$3\cdot \KK_{\mathcal{A}}(Q_i) - \log Q_i(x) \le  3\cdot \KK_{\mathcal{A}}(R) - \log R(x).$$
Therefore:
$$ 3a \le 3a - 2\log m  + O(\log b) \text{ and hence}$$
$$\log m \le O(\log b) = O(\log (\KK_{\mathcal{A}}(P) + D).$$
That is proved our claim.

{\bf Part $2$}.
To complete the proof we do the same things as in the proof of Theorem~\ref{m-l_sum}. 

If $x=x_1\ldots x_n$ is a prefix of $\omega$ and $Q$ is the best distribution for $x$ then by Lemma~\ref{lem1} 
$$\sum_{i=1}^{n-1} (Q(0| x_1 \ldots x_i) - P(0| x_1 \ldots x_i))^2 = O(\KK_{\mathcal{A}}(P) + D + \KK(P,Q) =O(\KK_{\mathcal{A}}(P) +D ).$$
(In the last equation we use $\KK(P,Q) = O(\KK(P)+\KK(Q)) = O(\KK_{\mathcal{A}}(P)+\KK_{\mathcal{A}}(Q))$.)
So, every probability distribution that is the best for some $x$ ``contributes'' $O(\KK_{\mathcal{A}}(P) + D)$ in the sum $\sum_{\textup{$x$ is a prefix of $\omega$}}  (H_{\mathcal{A}}(0| x) - P(0| x))^2$. There are $\poly(D + \KK_{\mathcal{A}}(P))$ such distributions, so we obtain the required estimation. 

\end{proof}
\begin{proof}[Proof of Theorem~\ref{sum_solomonoff_enum}]
The proof is the same as the proof of Theorem~\ref{sum_solomonoff} but with using Theorem~\ref{enum_m_l} instead of Theorem~\ref{m-l_sum}.
\end{proof}

\section{Open questions}
A natural question arises: can we get a better estimation in the last theorem than  $O(\KK(P) 2^{\frac{3 \KK(P)}{2}} )$? The exponential (in $\KK(P)$) estimation arises from our attempt to estimate the number of distributions that are optimal for some $x$. However, the author is not aware of  an example of $P$-random sequence $\omega$ such that there are exponentially many (in terms of $\KK(P)$ and $d(\omega|P)$) different best distributions for prefixes of $\omega$. 

Algorithmic statistics \cite{gtv,vs,suv} studies good distributions for strings among distributions on finite sets.  There exists a family of ``standard statistics'' that cover all the best distributions for finite strings. It is interesting: are these the same for distributions on the binary tree?
\section*{Acknowledgements}
I would like to thank Alexander Shen and Nikolay Vereshchagin  for their useful
discussions, advice and remarks.

Funded by the European Union (ERC, HOFGA, 101041696). Views and opinions expressed are however those of the author(s) only and do not necessarily reflect those of the European Union or the European Research Council. Neither the European Union nor the granting authority can be held responsible for them.
Also supported by FCT through the LASIGE Research Unit, ref.\ UIDB/00408/2020 and ref.\ UIDP/00408/2020.%
%
%
%
%
%
%

\end{document}